\documentclass[12pt]{article}
\usepackage{amsmath,amssymb,amsthm,mathtools}
\usepackage{graphicx,psfrag,epsf}
\usepackage{enumerate}
\usepackage{natbib}
\usepackage{hyperref}
\usepackage{url}
\usepackage{float}
\usepackage{listings}
\usepackage{subcaption}
\usepackage[T1]{fontenc}
\newcommand{\bs}[1]{\boldsymbol{#1}}
\newcommand{\LR}{\textsc{LR}}
\newtheorem{assumption}{Assumption}
\newtheorem{theorem}{Theorem}
\newtheorem{corollary}{Corollary}
\newtheorem{remark}{Remark}
\theoremstyle{remark}

\newcommand{\blind}{0}

\begin{document}

\bibliographystyle{plainnat}

\def\spacingset#1{\renewcommand{\baselinestretch}%
{#1}\small\normalsize} \spacingset{1}


\if0\blind
{
  \title{\bf Randomization inference for treatment effects on survival outcomes}
  \author{Lucy D'Agostino McGowan\thanks{
   \textbf{Funding/Support} Research reported in this publication was funded through a Patient-Centered Outcomes Research Institute (PCORI) Award (ME-2023C2-33433). \textbf{Role of Sponsor} The statements in this publication are solely the responsibility of the authors and do not necessarily represent the views of the Patient-Centered Outcomes Research Institute (PCORI), its Board of Governors or Methodology Committee.}\hspace{.2cm}\footnote{mcgowald@wfu.edu}\\
    Department of Statistical Sciences, Wake Forest University\\
     \\
    Joseph Rigdon\footnotemark[1] \\ 
    Department of Biostatistics and Data Science,\\ Wake Forest University School of Medicine \\
    \\
    Xinran Li \\
    Department of Statistics, University of Chicago \\
     \\
    Dylan Small \\
    Department of Statistics and Data Science,\\  The Wharton School, University of Pennsylvania \\
   }
  \maketitle
  \newpage
} \fi

\if1\blind
{
  \bigskip
  \bigskip
  \bigskip
  \begin{center}
    {\bf Randomization inference for treatment effects on survival outcomes}
\end{center}
  \medskip
} \fi

\begin{abstract}
The log-rank test and Kaplan--Meier plot are standard tools for analyzing time-to-event data in randomized clinical trials, yet neither provides a summary of the magnitude of the treatment effect. Practitioners typically fill this gap by reporting a hazard ratio from a Cox proportional-hazards model or an acceleration factor from an accelerated failure time (AFT) model, but both require assumptions beyond those needed for the log-rank test or Kaplan--Meier estimator. We propose two nonparametric confidence intervals for scalar effect-size summaries, an additive shift $c$ and a multiplicative factor $\rho$, obtained by inverting the log-rank test under sharp null hypotheses of constant treatment effects. Building on the randomization-inference framework of \citet{li2023randomization}, both intervals are valid under the randomization distribution alone, requiring no assumptions for the event-time distribution. We evaluate the proposed multiplicative interval via simulation, finding that it maintains nominal coverage across a range of censoring rates and sample sizes, including under data-generating processes that misspecify a parametric AFT model, while incurring only a modest efficiency loss compared to parametric AFT inference under correct specification. We illustrate the approach using data from a randomized trial of rhDNase for cystic fibrosis and provide \textsf{R} code and a Shiny application for ease of implementation.
\end{abstract}

\noindent%
\vfill

\newpage
\spacingset{1.45} 
\section{Introduction}
\label{sec:intro}

The Kaplan--Meier plot and the log-rank test are common tools used to analyze time-to-event data in randomized clinical trials. The Kaplan--Meier plot conveys the full shape of each arm's survival curve, while the log-rank test provides a $p$-value for the null hypothesis of no treatment effect. Neither, however, directly quantify the overall magnitude of the treatment effect, an important quantity when trying to establish both statistical and clinical significance. Practitioners typically fill this gap by reporting a hazard ratio from a Cox proportional-hazards model \citep{cox1972regression}. The Cox hazard ratio is well understood, but it relies on the proportional-hazards assumption, which may be violated in practice and which is not required for the validity of either the Kaplan--Meier estimator or the log-rank test. Accelerated failure time (AFT) models \citep{wei1992accelerated} offer an alternative effect-size summary, but likewise require assumptions about the data-generating process.

A potential alternative is to report effect-size summaries that, like the log-rank test itself, require no parametric model and are valid under the randomization distribution alone. \citet{li2023randomization} established a finite-population randomization-inference foundation for the log-rank test under noninformative censoring, showing that the log-rank statistic is asymptotically standard normal conditional on the potential outcomes under the sharp null of no treatment effect. Their framework accommodates Bernoulli randomized designs and imposes only mild regularity conditions on the joint distribution of event and censoring times. Building on \citet{li2023randomization}, we observe that their result extends to a family of sharp null hypotheses of constant treatment effects, both additive and multiplicative. Typically, randomization-based inference for a sharp null hypothesis involves transforming the potential outcomes \citep{rosenbaum2020design}; we extend this transformation to the potential censoring times as well, which ensures that the realized time at risk and event indicator are fully determined under the null. Inverting the resulting tests yields confidence intervals (CIs) for two complementary scalar summaries of the treatment effect. The first is an additive shift parameter, $c$. Under the null, $H_c$, treatment adds $c$ time units to every unit's event time. The CI for $c$ is expressed on the original time scale and answers the question ``by how much did treatment delay the event?'' The second is a multiplicative factor, $\rho$. Under the null, $H_\rho$, treatment multiplies every unit's event time by $\rho$. This parameter targets the same estimand as the acceleration factor in AFT models, but unlike model-based AFT inference, its validity derives from the randomization distribution alone, requiring no assumption on the event time distribution or the censoring beyond that censoring times are independent and identically distributed. Both CIs are computed by evaluating the log-rank $p$-value over a fine grid of candidate parameter values and retaining those not rejected at the $\alpha$-level. 

The remainder of the paper is organized as follows. Section~\ref{sec:setup} introduces the potential outcomes setup and assumptions from \citet{li2023randomization}. Section~\ref{sec:tests} derives the tests for sharp nulls of additive and multiplicative constant effects. Section~\ref{sec:ci} presents the CI construction and discusses computation and interpretation. Section~\ref{sec:sims} includes simulation results demonstrating the operating characteristics in comparison to parametric model-based methods. Section~\ref{sec:example} applies the method to real trial data and provides \textsf{R} code for ease of implementation. Finally, Section~\ref{sec:discussion} concludes with discussion of extensions and limitations.

\section{Setup, Notation and Assumptions}
\label{sec:setup}

We first introduce the setup, notation and assumptions from \citet{li2023randomization}. We consider an experiment on $n$ units with two treatment arms. 
For each unit $i$, let $T_i(1)$ and $T_i(0)$ denote the potential event times under treatment and control, and $C_i(1)$ and $C_i(0)$ denote the potential censoring times under treatment and control, respectively.
Let $\bs{T}(1)=\bigl(T_1(1),\ldots,T_n(1)\bigr)^\top$, $\bs{T}(0)=\bigl(T_1(0),\ldots,T_n(0)\bigr)^\top$, $\bs{C}(1)=\bigl(C_1(1),\ldots,C_n(1)\bigr)^\top$, and $\bs{C}(0)=\bigl(C_1(0),\ldots,C_n(0)\bigr)^\top$.

Let $Z_i\in\{0,1\}$ be the treatment assignment for unit $i$, where $Z_i=1$ if unit $i$ receives the active treatment and $Z_i=0$ otherwise. 
The observed  time at risk and the observed event indicator is 
\begin{align*}
	W_i = 
	\begin{cases}
		\min\{T_i(1), C_i(1)\}, & \text{if } Z_i = 1,\\
		\min\{T_i(0), C_i(0)\}, & \text{if } Z_i = 0, 
	\end{cases}
	\text{ and } 
	\Delta_i = 
	\begin{cases}
		1\{ T_i(1) \le C_i(1) \}, & \text{if } Z_i = 1,\\
		1\{ T_i(0) \le C_i(0) \}, & \text{if } Z_i = 0. 
	\end{cases}
\end{align*}

\begin{assumption}
$
Z_i \mid \bigl( \bs{T}(1),\bs{T}(0),\bs{C}(1),\bs{C}(0)\bigr) \overset{\mathrm{i.i.d.}}{\sim} \mathrm{Bern}(p_1)$, $1\le i\le n,$
for some $p_1\in(0,1)$.
\end{assumption}
In other words, conditional on all potential event and censoring times, treatment assignments are i.i.d. across units, each with probability $p_1$. This formalizes randomization, meaning treatment assignment is independent of both potential outcomes and potential censoring times.

Second, we consider the distribution of $(\bs{C}(1),\bs{C}(0))$, also called the censoring mechanism.
We focus on the case of noninformative i.i.d.\ censoring.

\begin{assumption}
\label{a2}
The potential censoring times for all units are independent of the potential event times for all units, that is, $(\bs{C}(1),\bs{C}(0)) \perp\!\!\!\perp (\bs{T}(1),\bs{T}(0))$, and the $n$ two-dimensional random vectors
\[
\bigl(C_1(1),C_1(0)\bigr),\ldots,\bigl(C_n(1),C_n(0)\bigr)
\]
are i.i.d.
\end{assumption}
This holds, for example, when censoring is due to administrative end of study or when subjects withdraw from the study for reasons unrelated to their outcome status. It would be violated if, for instance, subjects in worse health were more likely to drop out.

\section{Tests for the sharp null of constant effects}
\label{sec:tests}

\subsection{Additive shift}

Consider the null hypothesis
\[
H_c:\ T_i(1)=T_i(0)+c
\]
for some $c\ge 0$.

\begin{remark}
\label{remark:add}
	We can also consider the sharp null with constant effect $c<0$, by switching the treatment and control labels.
\end{remark}

Let
$
W_{i,c}=W_i+c( 1-Z_i )
$
for $ 1\le i\le n,$
and let $\LR_c$ denote the logrank statistic computed using the data
$
\{(Z_i,\Delta_i,W_{i,c}):1\le i\le n\}.
$

Define, for $1\le i\le n$, 
\[
\tilde T_i(1)=T_i(1),\qquad \tilde T_i(0)=T_i(0)+c,\qquad
\tilde C_i(1)=C_i(1),\qquad \tilde C_i(0)=C_i(0)+c,
\]

\begin{theorem}
\label{thm:additive}
Under the null hypothesis $H_c$ for any $c\ge 0$, if Assumptions 1 and 2 hold, and the regularity conditions (i.e., Condition 1 or 2 in \citet{li2023randomization}) hold for the transformed potential event times
$
\{(\tilde T_i(1),\tilde T_i(0))\}_{i=1}^n
$
and the transformed potential censoring times
$
\{(\tilde C_i(1),\tilde C_i(0))\}_{i=1}^n,
$
then
\[
\LR_c \mid \bs{T}(1),\bs{T}(0) \xrightarrow{d} N(0,1).
\]
\end{theorem}

\begin{proof}
First, we can verify that Assumptions 1 and 2 still hold after replacing the original potential event and censoring times by the  transformed potential event and censoring times. 

Second, under $H_c$, Fisher's null of no treatment effect holds for the transformed potential event times. 
That is, $\tilde T_i(1)=\tilde T_i(0)$ for $1\le i\le n$.

Third, with transformed potential event
and censoring times, 
for each unit $i$, the corresponding observed time at risk becomes $W_{i,c}$, and the observed event indicator is still $\Delta_i$.

From the above, the theorem follows from Theorems 4 and 5 in \citet{li2023randomization} applied to the transformed potential event and censoring times.
\end{proof}

\begin{remark} 
Standard randomization inference transforms potential outcomes under a sharp null. We extend this by also transforming potential censoring times. This ensures the realized time at risk and event indicator are fully determined under the null, as the observed value $\min\{T_i(z), C_i(z)\}$ depends on both. 
\end{remark}

\subsection{Multiplicative factor}
\label{sec:mult}

A natural alternative to the additive shift is to suppose that the treatment scales every potential event time by a common factor.  Define the sharp
null hypothesis
\begin{equation}
  H_\rho:\ T_i(1)=T_i(0)\rho,\quad 1\le i\le n,
  \label{eq:Hrho}
\end{equation}
for some $\rho>0$.  When $\rho>1$ the treatment prolongs survival (event times are longer under treatment); when $\rho<1$ it shortens survival.

Under $H_\rho$ we rescale the control-arm observed times by $\rho$. Specifically, let
\[
  W_{i,\rho} =
  \begin{cases}
    W_i, & Z_i=1,\\
    W_i\rho, & Z_i=0,
  \end{cases}
\]

and let $\LR_\rho$ denote the log-rank statistic computed from 
$\{(Z_i,\Delta_i,W_{i,\rho}):1\le i\le n\}$. Define, for $1\le i\le n$,
\[
  \tilde T_i(1)=T_i(1),\qquad \tilde T_i(0)=T_i(0)\rho,\qquad
  \tilde C_i(1)=C_i(1),\qquad \tilde C_i(0)=C_i(0)\rho.
\]

\begin{theorem}
\label{thm:multiplicative}
Under $H_\rho$ for any $\rho>0$, if Assumptions~1 and~2 hold and the regularity conditions of \citet{li2023randomization} hold for the transformed potential event times $\{(\tilde T_i(1),\tilde T_i(0))\}_{i=1}^n$ and the transformed potential censoring times $\{(\tilde C_i(1),\tilde C_i(0))\}_{i=1}^n$, then
\[
  \LR_\rho \mid \bs{T}(1),\bs{T}(0) \xrightarrow{d} N(0,1).
\]
\end{theorem}

\begin{proof}
The proof mirrors that of Theorem~\ref{thm:additive}.

First, Assumptions~1 and~2 still hold after replacing the original potential event and censoring times by the transformed potential event and censoring times. 

Second, under $H_\rho$, Fisher's null of no treatment effect holds for the transformed potential event times. That is, $\tilde T_i(1)=\tilde T_i(0)$ for $1\le i\le n$. 

Third, with the transformed potential event and censoring times, for each unit $i$, the corresponding observed time at risk becomes $W_{i,\rho}$, and the observed event indicator is still $\Delta_i$. 

From the above, the theorem follows from Theorems~4 and~5 of \citet{li2023randomization} applied to the transformed potential event and censoring times.
\end{proof}

\begin{remark}
\label{remark:mult}
The parameter $\rho$ in $H_\rho$ is a distribution-free time-acceleration factor. $\rho > 1$ means treatment multiplies every unit's event time by $\rho$, prolonging survival, while $\rho < 1$ shortens it. Unlike estimates from parametric AFT models, $\rho$ requires no distributional assumption and is valid under randomization alone. 
\end{remark}

\section{Construction}
\label{sec:ci}
An $\alpha$-level test of $H_c$ (and $H_\rho$, respectively) is obtained by rejecting when $|\LR_c|>z_{\alpha/2}$ (or $|\LR_\rho|>z_{\alpha/2}$), where $z_{\alpha/2}$ is the upper $\alpha/2$ quantile of a standard normal distribution. By test inversion, a $(1-\alpha)$ confidence set for the true additive shift $c^*$ is
\[
  \mathcal{C}_{\mathrm{add}} = \bigl\{c: |\LR_c| \le z_{\alpha/2}\bigr\},
\]
where for $c < 0$, $\LR_c$ is computed after switching treatment and control labels, as discussed in Remark~\ref{remark:add}. A $(1-\alpha)$ confidence set for the true multiplicative factor $\rho^*$ is
\[
  \mathcal{C}_{\mathrm{mult}} = \bigl\{\rho > 0 : |\LR_\rho| \le z_{\alpha/2}\bigr\}.
\]

\begin{corollary}
\label{cor:coverage}
Under the assumptions of Theorems~\ref{thm:additive} and~\ref{thm:multiplicative}, the confidence sets $\mathcal{C}_{\mathrm{add}}$ and $\mathcal{C}_{\mathrm{mult}}$ have asymptotic coverage at least $1-\alpha$.
\end{corollary}

\subsection{Computation}

In practice the confidence sets are computed by evaluating the log-rank $p$-value over a fine grid of candidate parameter values and retaining those for which the $p$-value exceeds $\alpha$.
For the additive shift, one searches over a grid
$c\in\{c_{\min},\ldots,c_{\max}\}$; for the multiplicative factor one searches over a grid on the log scale, $\rho\in\{\exp(\ell):\ \ell\in\{\ell_{\min},\ldots,\ell_{\max}\}\}$.

Point estimates are obtained as the value of $c$ (or $\rho$) that maximizes the log-rank $p$-value, i.e., the value most consistent with the observed data.

\subsection{Interpretation}
\label{sec:interp}

The additive shift $c$ is expressed on the original time scale and quantifies the number of time units by which treatment delays the event. The multiplicative factor $\rho$ quantifies the factor by which treatment scales event times, so that $\rho > 1$ indicates prolonged survival and $\rho < 1$ indicates shortened survival. The multiplicative factor $\rho$ is the same estimand as the acceleration factor in an AFT model, but unlike model-based AFT inference, its validity derives from the randomization distribution alone. 

Both parameters have a natural connection to the Kaplan--Meier plot. Under $H_c$, shifting the control arm's time axis right by $c$ should approximately recover the treatment arm's curve. Under $H_\rho$, stretching the control arm's time axis by $\rho$ should do the same. Overlaying the transformed control curve on the Kaplan--Meier plot provides a visual check of whether a constant-effect summary is adequate for the data at hand.

\subsubsection{Connection to the Hazard Ratio}

The hazard ratio from a Cox proportional hazards model operates on the hazard scale \citep{cox1972regression}, while the multiplicative factor $\rho$ operates on the time scale. Since the hazard ratio is so commonly reported in practice, it may be useful to draw a connection between the two. The multiplicative factor $\rho$ estimates the factor by which treatment multiplies event times, whereas the hazard ratio estimates the factor by which treatment multiplies the instantaneous event rate. While these quantities summarize different aspects of the treatment effect, a direct relationship arises under a Weibull model. Suppose event times follow a Weibull AFT model, so that under control
\[
\log(T(0))=\mu+\sigma\varepsilon,
\]
where $\varepsilon$ has an extreme value distribution, implying that $T(0)\sim \mathrm{Weibull}(\gamma = 1/\sigma, \lambda =\exp(-\mu))$ \citep{kalbfleisch2002statistical}. Under a constant multiplicative treatment effect as in \eqref{eq:Hrho}, $T(1)=T(0)\rho$. The baseline hazard is
\begin{equation}
h_0(t)=\lambda\gamma(\lambda t)^{\gamma-1},
\label{eq:h0}
\end{equation}
and under the time rescaling $T(1)=T(0)\rho$, the hazard under treatment is then
\begin{equation}
h_1(t)=\rho^{-1}h_0(t/\rho).
\label{eq:h1}
\end{equation}
Substituting the Weibull form from \eqref{eq:h0} in \eqref{eq:h1} yields
\[
h_1(t)=\rho^{-\gamma}h_0(t).
\]
Thus, the hazard ratio is $\mathrm{HR}=h_1(t)/h_0(t)=\rho^{-\gamma}=\rho^{-1/\sigma}$. In the special case where $T(0)$ follows an exponential distribution, $\gamma = 1$ and this simplifies to $\mathrm{HR} = 1/\rho$ so that the hazard ratio is simply the reciprocal of the acceleration factor. In general, the acceleration factor and hazard ratio summarize different, though related, aspects of the treatment effect.

\section{Simulations}
\label{sec:sims}

We conducted a simulation study to evaluate the finite-sample operating characteristics of the proposed randomization-based confidence intervals for the multiplicative effect parameter $\rho$, and to compare their performance with standard parametric AFT model-based inference.

We considered two data-generating processes to compare the approaches when the parametric model was correctly specified and misspecified. In the first data-generating mechanism, baseline event times were drawn from a Weibull distribution,
\[
T(0) \sim \mathrm{Weibull}(\gamma = 1.5, \lambda = 100),
\]
corresponding to a setting in which the Weibull AFT model is correctly specified. In the second mechanism, baseline event times were drawn from a log-logistic distribution,
\[
T(0) \sim \mathrm{LogLogistic}(\alpha = 4, \beta = 100),\]
which induces misspecification of the Weibull AFT likelihood.

In both cases, treatment assignment was generated as $Z \sim \mathrm{Bernoulli}(0.5)$, and a constant multiplicative treatment effect was imposed via $T(1) = T(0)\rho$. We considered two values of the treatment effect, $\rho \in \{1.0, 1.25\}$, where $\rho = 1$ corresponds to the null hypothesis of no treatment effect and $\rho = 1.25$ represents a beneficial effect that increases event times under treatment. Independent right censoring was generated from a uniform distribution on $[0, \theta]$, where $\theta$ was calibrated separately for each scenario to achieve target censoring rates of approximately 20\%, 50\%, and 80\%.

For each simulated dataset, we constructed confidence intervals using two approaches. The first is the proposed nonparametric randomization-based method, obtained by inverting the log-rank test over a fine grid of candidate values of $\rho$. The second is a parametric approach based on a Weibull AFT model fit by maximum likelihood, with Wald-type confidence intervals constructed on the log scale and exponentiated to obtain intervals for $\rho$.

We considered sample sizes $n \in \{50, 200\}$. For each combination of data-generating process, sample size, censoring rate, and true parameter value, we generated 3,000 simulated datasets. Performance under the null ($\rho = 1$) was evaluated via Type I error, defined as the proportion of intervals failing to contain the true value. Performance under the alternative ($\rho = 1.25$) was evaluated via coverage probability and median interval width.

\subsection{Results}

Under the null, the randomization-based intervals maintained Type I error close to the nominal 5\% level across both data-generating processes, sample sizes, and censoring rates (Figure~\ref{fig:type1}). The Weibull AFT intervals also controlled Type I error well under correct specification, but exhibited inflated Type I error under the log-logistic data-generating process. Under the alternative, both methods achieved comparable coverage when the Weibull model was correctly specified (Figure~\ref{fig:coverage}), with the parametric approach yielding slightly narrower intervals on average (Figure~\ref{fig:width}). Under misspecification, the Weibull AFT intervals showed degraded coverage, whereas the randomization-based intervals remained close to the nominal 95\% level across all scenarios. Overall, the results highlight a trade-off between efficiency and robustness in that parametric AFT methods offer gains in precision under correct specification, while the proposed randomization-based procedure provides reliable inference across a broader class of data-generating mechanisms.

\begin{figure}[H]
  \centering
  \includegraphics[width=\textwidth]{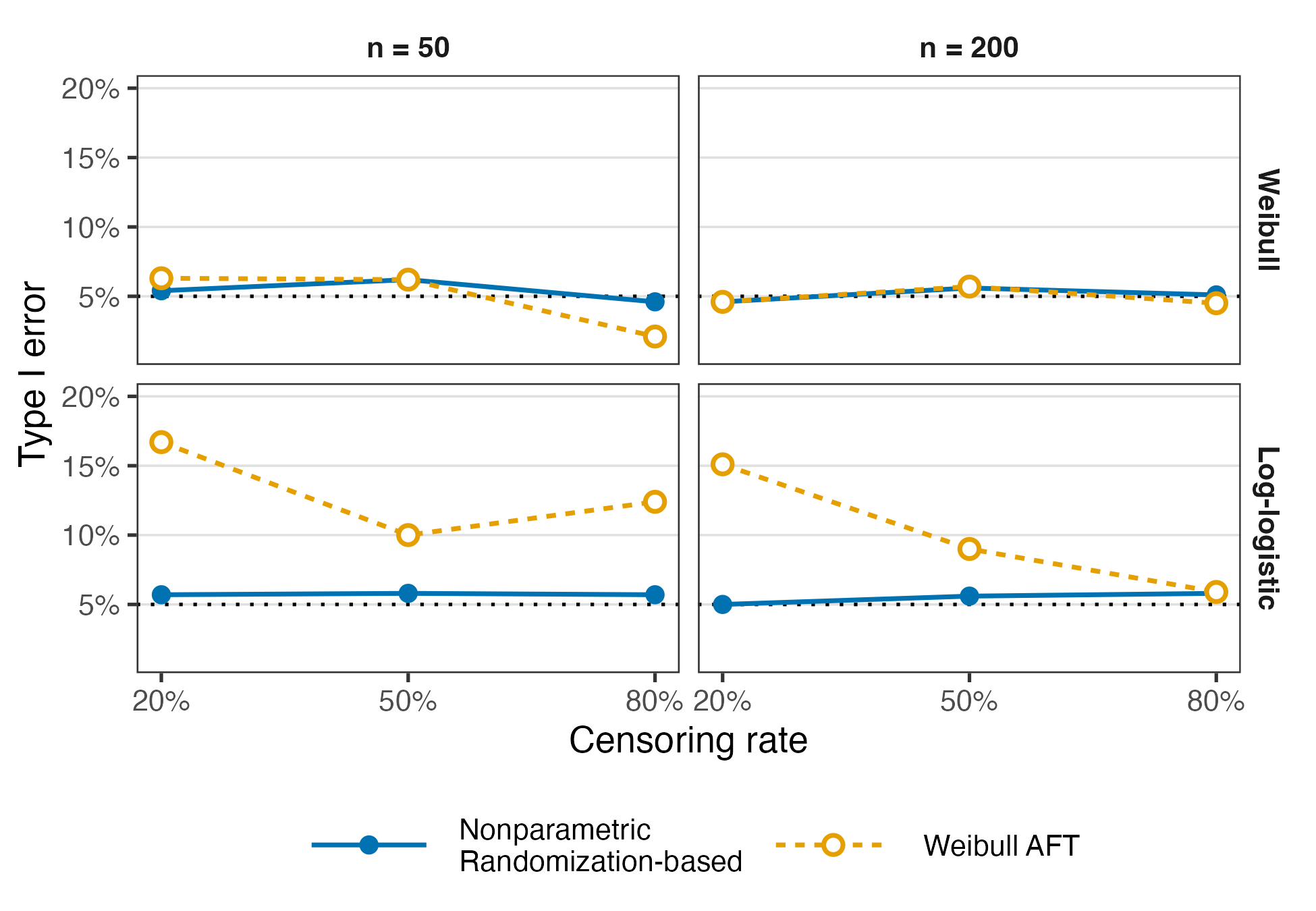}
  \caption{Type I error of 95\% confidence intervals for the randomization-based and Weibull AFT methods under the null ($\rho = 1$), across censoring rates (20\%, 50\%, 80\%), sample sizes ($n \in \{50, 200\}$), and data-generating processes based on 3,000 simulations per scenario. The dotted horizontal line marks the nominal 5\% level.}
  \label{fig:type1}
\end{figure}

\begin{figure}[H]
  \centering
  \includegraphics[width=\textwidth]{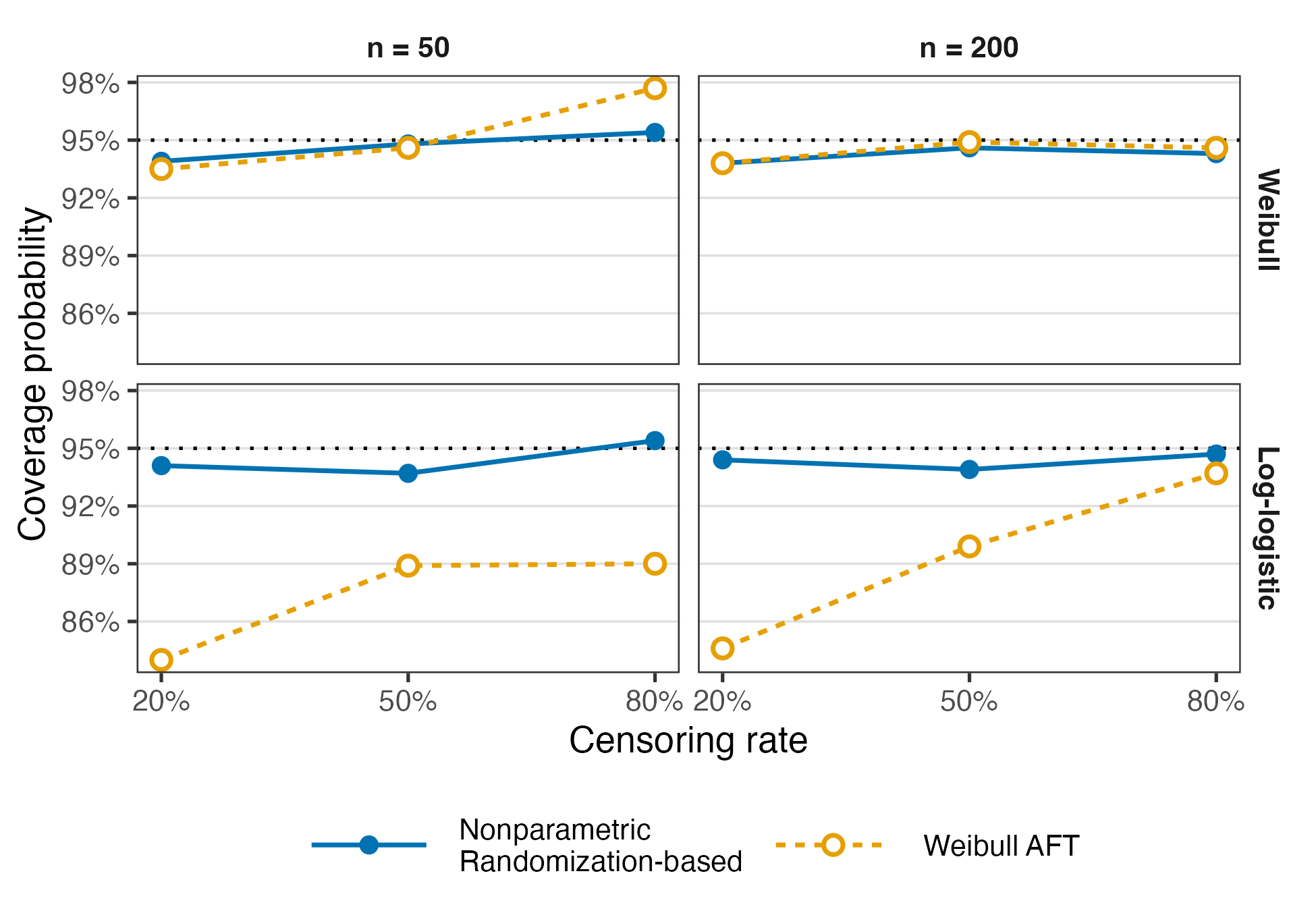}
  \caption{Empirical coverage of 95\% confidence intervals for the randomization-based and Weibull AFT methods under the alternative ($\rho = 1.25$), evaluated across censoring rates (20\%, 50\%, 80\%), sample sizes ($n \in \{50, 200\}$), and data-generating processes based on 3,000 simulations per scenario. The dotted horizontal line marks the nominal 95\% level.}
  \label{fig:coverage}
\end{figure}

\begin{figure}[H]
  \centering
  \includegraphics[width=\textwidth]{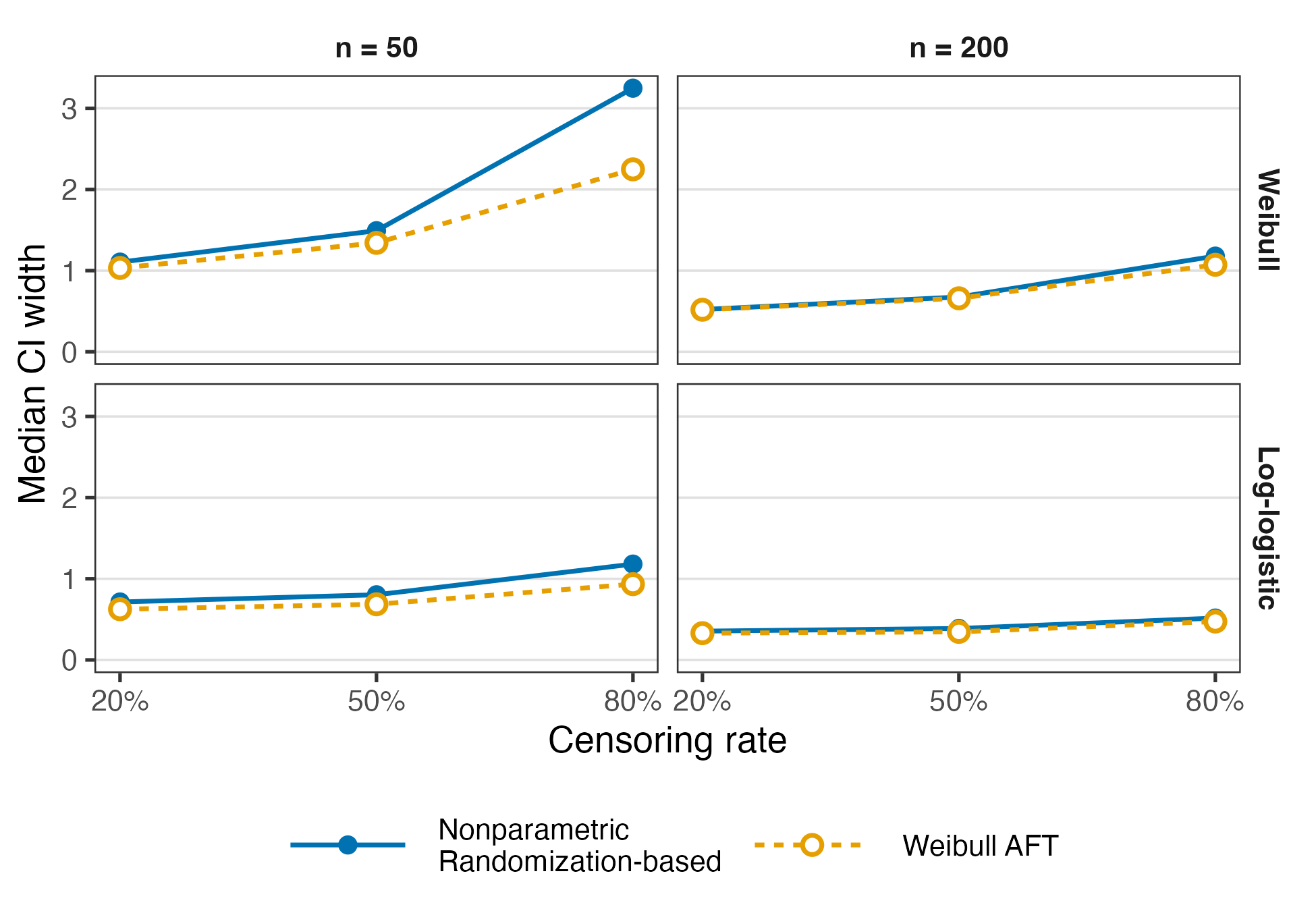}
  \caption{Median width of 95\% confidence intervals for the randomization-based and Weibull AFT methods under the alternative ($\rho = 1.25$), evaluated across censoring rates (20\%, 50\%, 80\%), sample sizes ($n \in \{50, 200\}$), and data-generating processes based on 3,000 simulations per scenario.}
  \label{fig:width}
\end{figure}

\section{Real Data Example}
\label{sec:example}

We illustrate the proposed CIs using data from a randomized trial of recombinant human deoxyribonuclease~I (rhDNase) for the treatment of cystic fibrosis \citep{therneau1997rhdnase}. Here, the primary endpoint is time to the first pulmonary exacerbation. The dataset is available as \texttt{rhDNase} in the \texttt{survival} package for \textsf{R} \citep{survival-book, survival-package}. Of the $n=647$ enrolled subjects, we use the $n=641$ with positive observed times (six subjects had an infection at enrollment and are excluded). Code~\ref{lst:setup} provides \textsf{R} code to prepare the dataset.

\begin{lstlisting}[label=lst:setup, 
                   caption={\textsf{R} code to prepare \texttt{rhDNase} dataset.}]
library(survival)
dat_first <- subset(rhDNase, !duplicated(id))
dat <- data.frame(
  Y = ifelse(!is.na(dat_first$ivstart), dat_first$ivstart,
                 as.numeric(dat_first$end.dt - dat_first$entry.dt)),
  Z = dat_first$trt,
  event = as.integer(!is.na(dat_first$ivstart))
)
dat <- dat[dat$Y > 0, ] # exclude subjects infected at enrollment
\end{lstlisting}

Of the 641 subjects, 323 were randomized to placebo ($Z=0$) and 318 to rhDNase ($Z=1$), with 138 and 103 events respectively. The log-rank test gives $p=0.0056$, providing evidence that rhDNase delays pulmonary exacerbation.

We apply the test-inversion procedure of Section~\ref{sec:ci} at the $\alpha=0.05$ level, searching over a grid of 1001 values of $c \in [-80, 80]$ days for the additive shift. The additive point estimate is $50.1$ days with 95\% CI $[17.1,\, 73.0]$ days, indicating that rhDNase delayed the time to first exacerbation by between roughly two and ten weeks, with a point estimate of about 50 days. Code~\ref{lst:add} includes an \textsc{R} function to estimate the additive effect.

\begin{lstlisting}[label=lst:add, 
                   caption={\textsf{R} function to estimate the additive treatment effect}]
ci_additive <- function(data, time, treatment, event, 
                        lower, upper, alpha = 0.05, ngrid = 1001) {
  grid_c <- seq(lower, upper, length.out = ngrid)
  
  levels   <- unique(data[[treatment]])
  trt_label  <- levels[1]
  ctrl_label <- levels[2]
  
  pval <- vapply(grid_c, function(c) {
    t <- data[[time]]
    z <- data[[treatment]]
    e <- data[[event]]
    
    # If c <= 0, flip treatment labels
    if (c <= 0) {
      z <- ifelse(z == trt_label, ctrl_label, trt_label)
      c <- -c
    } 
    
    t[z == 0] <- t[z == 0] + c
    survdiff(Surv(t, e) ~ z, rho = 0)$pvalue
  }, numeric(1))
  
  data.frame(
    est   = grid_c[which.max(pval)],
    lower = min(grid_c[pval > alpha]),
    upper = max(grid_c[pval > alpha])
  )
}
ci_additive(dat, "Y", "Z", "event", -80, 80)
#|     est lower upper
#| 1 50.08 17.12 72.96
\end{lstlisting}

To estimate the multiplicative factor, we apply the test-inversion procedure of Section~\ref{sec:ci} at the $\alpha=0.05$ level and search over a grid of 1001 values of $\rho$ on a log scale over $[e^{-3}, e^{3}]$. The multiplicative point estimate $\hat\rho = 1.55$ with 95\% CI $[1.14, 1.89]$ suggests that rhDNase increases the time to event by a factor of $1.55$. In other words, treatment is estimated to extend time to event by about 55\% on a multiplicative time scale. Code~\ref{lst:mult} includes an \textsc{R} function to estimate the multiplicative effect.

\begin{lstlisting}[label=lst:mult, 
                   caption={\textsf{R} function to estimate the multiplicative treatment effect}]
ci_multiplicative <- function(data, time, treatment, event,
                              log_lower, log_upper,
                              alpha = 0.05, ngrid = 1001) {
  grid_rho <- exp(seq(log_lower, log_upper, length.out = ngrid))
  
  pval <- vapply(grid_rho, function(rho) {
    t <- data[[time]]
    z <- data[[treatment]]
    e <- data[[event]]
    
    t[z == 0] <- t[z == 0] * rho
    survdiff(Surv(t, e) ~ z, rho = 0)$pvalue
  }, numeric(1))

  data.frame(
    est   = grid_rho[which.max(pval)],
    lower = min(grid_rho[pval > alpha]),
    upper = max(grid_rho[pval > alpha])
  )
}
ci_multiplicative(dat, "Y", "Z", "event", -3, 3)
#|     est lower upper
#| 1  1.55  1.14  1.89
\end{lstlisting}

As mentioned in Section~\ref{sec:interp}, the adequacy of a constant-effect summary can be assessed visually by transforming the control arm's Kaplan--Meier curve and checking whether it tracks the treatment arm. Figure~\ref{fig:fig-1} overlays two such transformations: the control curve shifted right by $\hat{c} = 50$ days (additive) and the control curve with its time axis stretched by $\hat\rho = 1.55$ (multiplicative), so that a step at time $t$ is replotted at $t + \hat{c}$ or $t\hat\rho$ respectively. The multiplicative rescaling tracks the treatment arm more closely across the follow-up period than the additive shift, supporting the use of $\hat\rho = 1.55$ as the preferred summary of the treatment effect. Code~\ref{lst:km} includes \textsf{R} code to create such a plot.

\begin{lstlisting}[label=lst:km, 
                   caption={\textsf{R} code to overlay the time-transformed control arm curves on the Kaplan--Meier plot.}]
library(ggsurvfit)
library(ggplot2)

# Main K-M plot
km_fit <- survfit2(Surv(Y, event) ~ Z, data = dat)

# Control arm data for transformations
km_control <- survfit2(Surv(Y, event) ~ 1, 
                       data = subset(dat, Z == 0)) |>
  tidy_survfit()

km_mult <- km_control |> transform(time = time * 1.55)
km_add  <- km_control |> transform(time = time + 50)

km_fit |>
  ggsurvfit() +
  geom_step(
    data = km_mult,
    aes(x = time, y = estimate),
    linetype = "dashed",
    inherit.aes = FALSE
  ) +
  geom_step(
    data = km_add,
    aes(x = time, y = estimate),
    linetype = "dotted",
    inherit.aes = FALSE
  ) +
  scale_ggsurvfit(x_scales = list(limits = c(0, 189)))
\end{lstlisting}

\begin{figure}
    \centering
    \includegraphics[width=0.75\linewidth]{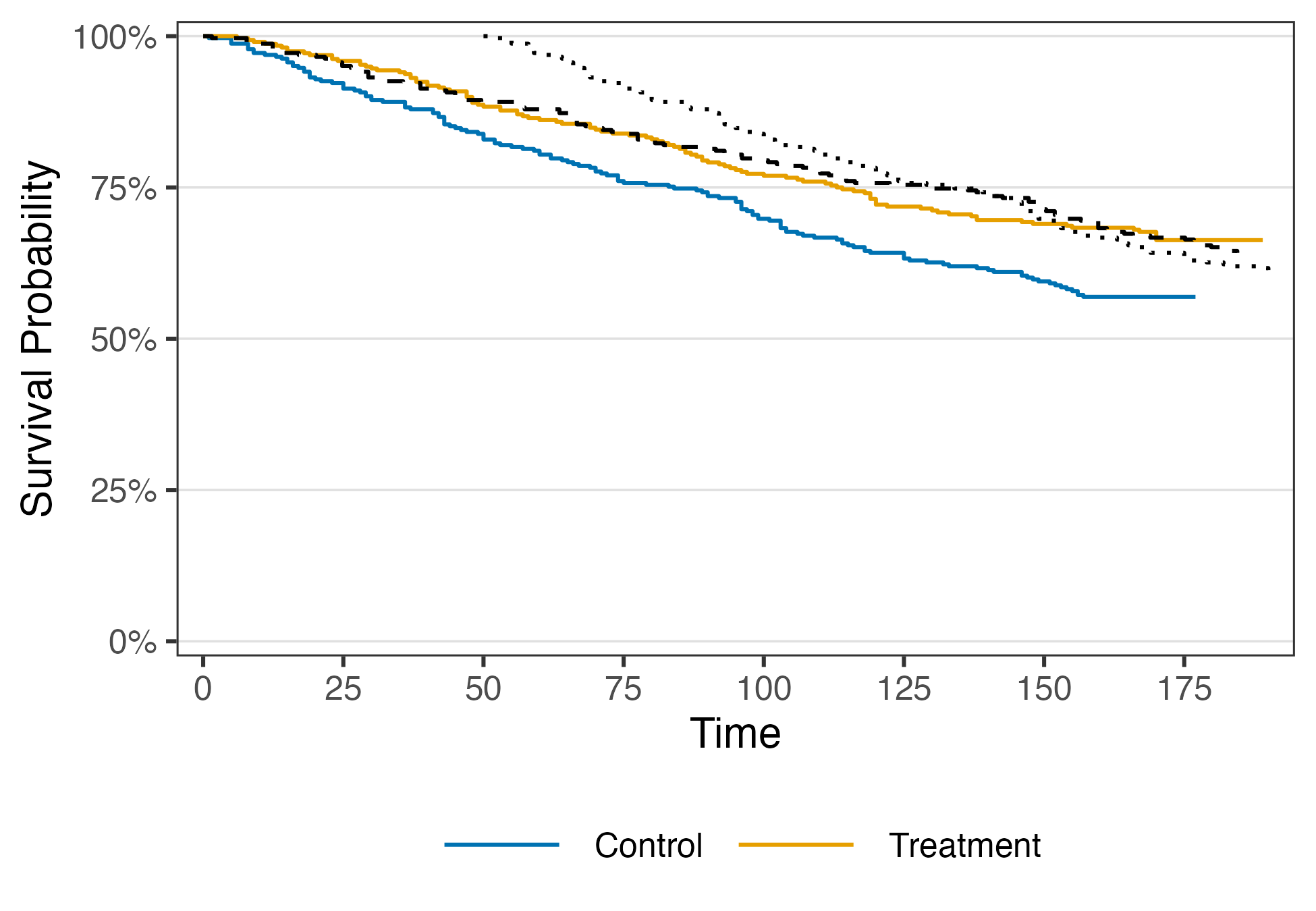}
    \caption{Kaplan--Meier survival curves for the rhDNase trial. The solid lines show the estimated survival functions for the control (blue) and rhDNase (orange) arms. The dashed line shows the control arm's curve with its time axis stretched by $\hat\rho = 1.55$ (multiplicative); the dotted line shows the control arm's curve shifted right by $\hat{c} = 50$ days (additive). The closer agreement of the multiplicative transformation with the treatment arm curve suggests that a constant time-acceleration factor is a more adequate summary of the treatment effect than a constant additive shift.}
    \label{fig:fig-1}
\end{figure}

For ease of implementation, we have created a Shiny application, available at: 
\url{https://lucy.shinyapps.io/survival_ci}.

\section{Discussion}
\label{sec:discussion}

We describe two nonparametric confidence intervals for summarizing treatment effects for randomized trials with time-to-event endpoints, one for an additive shift and one for a multiplicative factor. Both are obtained by inverting the log-rank test under sharp null hypotheses of constant effects, inheriting the randomization-validity established by \citet{li2023randomization} without requiring any model for the time-to-event distribution. 

A strength of the proposed approach is its conceptual alignment with existing nonparametric approaches used in practice. The log-rank test and Kaplan--Meier estimator are standard tools for survival analysis in randomized trials, yet neither provides a summary metric of the overall effect magnitude. The additive and multiplicative CIs proposed here are computed directly from the same log-rank statistic already often reported, require no additional distributional assumptions, and remain valid under the randomization distribution alone.

A related test-inversion approach was proposed by \citet{lin2016confidence}, who build on work by \citet{peto1977design} to construct a confidence interval for the hazard ratio by inverting the partial-likelihood score test under the Cox proportional hazards model. Their method shares the property that the confidence interval excludes the null value if the log-rank test is significant, and it yields narrower and more accurate intervals than the standard Wald confidence interval. However, their procedure targets the hazard ratio, a parameter whose interpretation relies on the proportional hazards assumption, and validity depends on correct model specification. By contrast, the confidence intervals proposed here target an additive shift or multiplicative factor and require no assumptions for the event-time distribution. While both utilize the log-rank test, the distinction is that our approach requires fewer assumptions, deriving its validity from the randomization distribution alone.

Both parameters summarize the treatment effect by a single scalar and therefore are most meaningful when the effect is constant across time. In the presence of treatment effect heterogeneity, the Kaplan--Meier curves themselves along with measures such as restricted mean survival time may be more informative. Alternatively, the constant-effect assumption could be relaxed by considering subgroup analyses, should the subgroups be known.

A limitation is that the proposed CIs will generally be wider than their parametric counterparts when a parametric model is correctly specified. Because the log-rank statistic is rank-based, it discards the magnitude of the observed times and uses only their ordering. A correctly specified Weibull or log-normal AFT model, by contrast, exploits the full likelihood and will yield a more precise estimate of the acceleration factor $\rho$ under that model. This efficiency loss is a trade-off for robustness in that the nonparametric CIs are valid regardless of the event-time distribution. Practitioners who are confident in a parametric family may prefer the narrower intervals it affords.

Although our randomization-based justification requires assumptions on the censoring mechanism (specifically, noninformative i.i.d. censoring as stated in Assumption~\ref{a2}), it places no stochastic assumptions on the potential event times. In contrast, the classical asymptotic theory for the log-rank test treats the event times as stochastic and requires independent/noninformative censoring, while allowing more general censoring schemes than i.i.d. censoring \citep{gill1980censoring}. Consequently, even in settings where the censoring assumptions underlying our randomization-based analysis fail, the proposed procedure may still admit a valid classical large-sample justification under an appropriate superpopulation model for the event times. In this sense, the procedure exhibits a form of robustness to modeling assumptions, since the randomization-based and classical asymptotic justifications rely on different sources of stochastic structure and impose assumptions on different components of the data-generating process.

In summary, the proposed confidence intervals offer a nonparametric complement to the log-rank test and Kaplan--Meier plots for summarizing time-to-event treatment effects in randomized trials. They are easy to compute using standard \textsf{R} tools, directly interpretable on clinically meaningful scales, and grounded in the same randomization-validity framework as the log-rank test itself. We hope the accompanying Shiny application makes them accessible to a broad audience of applied researchers.

\section{Data Availability}

Code to reproduce all figures and results can be found on Github: \if1\blind \textit{URL removed for blinding}\fi\if0\blind\url{https://github.com/LucyMcGowan/2026-nonparametric-survival}\fi
.
\bibliography{bibliography}
\end{document}